\documentclass[a4paper,onecolumn,10pt,accepted=2022-02-08]{quantumarticle}
\pdfoutput=1

\usepackage[utf8]{inputenc}
\usepackage[english]{babel}
\usepackage[T1]{fontenc}
\usepackage{amsmath}
\usepackage{amsthm}
\usepackage{amssymb}
\usepackage{hyperref}
\usepackage{enumitem}
\usepackage{doi}
\usepackage[numbers,sort&compress]{natbib}

\newtheorem{theorem}{Theorem}
\newtheorem{lemma}[theorem]{Lemma}

\newtheorem{definition}[theorem]{Definition}
\newtheorem{fact}[theorem]{Fact}

\newtheorem{conjecture}{Conjecture}

\newtheorem*{theorem*}{Theorem}

\newcommand{\hs}{\mathcal{H}}
\newcommand{\dm}{\mathcal{D} (\hs_A\otimes\hs_B)}
\newcommand{\PPT}{PPT}
\newcommand{\SEP}{SEP}
\newcommand{\CC}{CC}
\newcommand{\PROD}{Prod}

\DeclareMathOperator{\Tr}{\mathrm{Tr}}
\def\identity{\leavevmode\hbox{\small1\kern-3.8pt\normalsize1}}
\newcommand{\pt}[1]{#1{}^{T_B}}
\newcommand{\pospart}[1]{{\left(#1\right)}_+}
\newcommand{\negpart}[1]{{\left(#1\right)}_-}
\newcommand{\half}{\frac{1}{2}}
\newcommand{\norm}[1]{\left\| #1 \right\|_1}
\newcommand{\ket}[1]{\left | #1 \right\rangle}
\newcommand{\bra}[1]{\left \langle#1 \right |}
\newcommand{\proj}[1]{\ket{#1}\bra{#1}}

\begin{document}

\title{Hierarchy of correlation quantifiers comparable to negativity}

\author[1,2]{Ray Ganardi}
\orcid{0000-0002-0880-7816}
\email{ray.ganardi@ug.edu.pl}

\author[3]{Marek Miller}
\orcid{0000-0002-1850-4499}

\author[1]{Tomasz Paterek}
\orcid{0000-0002-8490-3156}

\author[2]{Marek \.Zukowski}
\orcid{0000-0001-7882-7962}

\affil[1]{Institute of Theoretical Physics and Astrophysics, Faculty of Mathematics, Physics, and Informatics, University of Gdańsk, 80-308 Gdańsk, Poland}
\affil[2]{International Centre for Theory of Quantum Technologies (ICTQT), University of Gdańsk, 80-308 Gdańsk, Poland}
\affil[3]{Centre of New Technologies, University of Warsaw, 02-097 Warsaw, Poland}

\maketitle

\begin{abstract}
Quantum systems generally exhibit different kinds of correlations.
In order to compare them on equal footing, one uses the so-called distance-based approach where different types of correlations are captured by the distance to different sets of states.
However, these quantifiers are usually hard to compute as their definition involves optimization aiming to find the closest states within the set.
On the other hand, negativity is one of the few computable entanglement monotones, but its comparison with other correlations required further justification.
Here we place negativity as part of a family of correlation measures that has a distance-based construction.
We introduce a suitable distance, discuss the emerging measures and their applications, and compare them to relative entropy-based correlation quantifiers.
This work is a step towards correlation measures that are simultaneously comparable and computable.
\end{abstract}

\section{Introduction}

Correlations between quantum particles come in many different flavours.
In the simplest case of bipartite systems, we typically distinguish classical correlations, quantum discord, quantum entanglement, quantum steering, and Bell nonlocality~\cite{Adesso_2016}.
These correlations form a hierarchical structure~\cite{Jirakova_2021}, can be measured by a variety of quantifiers~\cite{Horodecki_2009a,Modi_2012}, which even for entanglement are generally not comparable as they lead to different ordering of entangled states~\cite{Virmani_2000,Miranowicz_2004}.
How do we then quantify discord and classical correlations in a way comparable to a given entanglement monotone?

A common answer to this question is to define the amount of correlation as a distance to a set of uncorrelated states~\cite{Modi_2010,Vedral_1997}, in the same spirit as monotones introduced in the framework of resource theories~\cite{Chitambar_2019}.
For example, the distance from a given state to the set of separable states quantifies the amount of quantum entanglement,
the distance to the so-called classical states defines quantum discord, and the distance to product states gives rise to total correlation.
These quantifiers are comparable because the same distance measure is used and one only modifies the set of uncorrelated states.
When the distance is equipped with an operational interpretation, the resulting correlation quantifier also inherits this operational interpretation.

A concrete example of this framework in action is provided by the hierarchy of correlations defined using relative entropy as the distance.
Although relative entropy is technically not a metric, it admits operational interpretation~\cite{Vedral_1997a} and leads to quantifiers known from information theory.
Namely, the distance to the product states is given by the mutual information.
Furthermore, quantum relative entropy is monotonic under completely positive trace preserving (CPTP) maps, making the resulting quantifiers correlation monotones automatically.

The problem is that such defined measures are hard to compute.
Relative entropy of entanglement has a closed expression only for a very limited number of cases~\cite{Horodecki_2009a,Miranowicz_2008a,Friedland_2011} and in general the closest distance to the set of separable states (as well as to the set of classical states) is hard to find.
The problem is not just with relative entropy --- in fact for a long time, the only computable entanglement monotone was negativity (see Sec. XV.F in Ref.~\cite{Horodecki_2009a} and comments on concurrence and $\kappa$-entanglement below).

Negativity is defined as the sum of negative eigenvalues after the state of a bipartite system is partially transposed~\cite{_yczkowski_1998,_yczkowski_1999,Vidal_2002,Plenio_2005,Lee_2000}.
Operationally, negativity features in bounds on various quantum information measures: an upper bound to distillable entanglement~\cite{Vidal_2002}, a lower \emph{and} an upper bound to entanglement cost under PPT operations~\cite{Audenaert_2003,Ishizaka_2004}, and a lower bound on entanglement dimension~\cite{Eltschka_2013}.
For a large class of states~\cite{Audenaert_2003} (including two qubit states~\cite{Ishizaka_2004}), the relation to entanglement cost under PPT operations is even exact.

The question we attack here is quantification of other correlations in a way comparable to negativity, so that statements such as ``entanglement is half of total correlations'' are meaningful.
We follow the distance-based approach and introduce the partial transpose distance which equals negativity for a broad class of states --- we conjecture based on extensive numerical evidence that the equality holds for any state.
We discuss properties of the hierarchy of quantifiers obtained using this distance and compare them to relative entropy-based quantifiers.
We also present exact expressions of the resulting quantifiers for certain classes of states and discuss their applications to detection of non-decomposability and non-Markovianity of dynamics.

\subsection*{Related works}

The issue of comparable correlations has been addressed through the distance approach in multiple works~\cite{Adesso_2016}.
For example, Refs.~\cite{Bromley_2014,Roga_2016} study the correlations measured by Bures and Hellinger distance and Ref.~\cite{Paula_2013} studies correlations measured by trace distance.
Although these studies are clearly useful, the resulting measures are hard to compute.
They could be estimated by sampling uncorrelated states and computing the distance between the state of interest and the sampled state.
This way we get upper bounds on correlation quantifiers and in principle the state of interest might not even be correlated.
The current work complements these approaches by providing an easily computable lower bound to entanglement, that is negativity.

Perhaps the most prominent example of the distance approach is given by relative entropy-based quantifiers.
In this context one might think that since negativity captures whether partially transposed state fails to be a state, a natural ``entanglement'' quantifier would be relative entropy to the set of states which admit positive partial transposition (PPT states)~\cite{Rains_2001,Rains_1999}.
Recently, an approximation of relative entropy by semidefinite program (SDP) was discovered~\cite{Fawzi_2018,Fawzi_2018a}, allowing the computation of relative entropy to PPT states.
The approach is based on a numerical approximation of the matrix logarithm, which allows the relative entropy to be expressed as an SDP\@.
The present work complements this by considering a different distance measure, for which the correlations are comparable to negativity.

Another general approach to defining a hierarchy of correlations is through the use of base norms~\cite{Regula_2017}.
Indeed, the relation between negativity and base norms was noticed since the introduction of negativity~\cite{Vidal_2002}.
Base norms are closely related to robustness~\cite{Vidal_1999a} and measures defined in this fashion satisfy desirable properties.
Unfortunately, they also suffer from the difficult optimization problem.
Additionally, they do not capture the degree of total correlation between different pure entangled states (see Appendix~\ref{app:robustness}).
The main issue is base norms are defined to quantify distance to convex sets, while the set of product states is manifestly non-convex.
Thus any hierarchy of correlations defined within this framework cannot include total correlation.

Ref.~\cite{Wang_2020} introduced a computable entanglement monotone called $\kappa$-entanglement, that quantifies the PPT exact entanglement cost.
In Ref.~\cite{Wang_2020b}, negativity and $\kappa$-entanglement were shown to be members of a family of entanglement measures called $\alpha$-negativity.
These measures are related to the sandwiched $\alpha$-Rényi relative entropy distance to the set of PPT states.
By taking this distance to other sets of uncorrelated states, one could define a family of comparable quantifiers for any $\alpha$.
However, it turns out that taking the limit $\alpha \to 1$ collapses all obtained quantifiers to one and the same value, given by the negativity.
In other words this approach does not differentiate between the different kinds of correlations in a quantum state.

Another computable entanglement monotone is concurrence~\cite{Hill_1997} and it is interesting to ask if correlation monotones on equal footing to concurrence could be defined via the distance approach.
Indeed concurrence is equal to the Hilbert-Schmidt distance for a restricted set of states~\cite{0712.1015v2}.
This suggests that quantifiers of other correlations could be defined with the help of the Hilbert-Schmidt distance.
In this way Refs.~\cite{Girolami_2011,Debarba_2012} show a relation between geometric discord and negativity.
However, Hilbert-Schmidt distance is well-known to be non-contractive under CPTP operations~\cite{Ozawa_2000,Piani_2012}.
This leads to difficulties in interpreting the resulting quantifiers as a measure of correlation.
In contrast, present work introduces the family of correlations using the distance that is contractive.

\subsection*{Notation}

We denote a Hilbert space by $\hs$, and $\mathcal{D}\left(\hs\right)$ denotes the set of states (density operators).
Given a bipartite system $\hs_A \otimes \hs_B$, we denote partial transpose on the computational basis of $\hs_B$ as $T_B$.
The set of all partially transposed states is denoted $PT$.
The trace norm of a matrix $A$ is defined as $\norm{A} = \Tr{\sqrt{A^\dagger A}}$.
Given a Hermitian matrix $A$, there exists a Jordan decomposition into its positive part and negative part, i.e. $A = \pospart{A} + \negpart{A}$, with $\pospart{A} > 0$, $\negpart{A} < 0$, and $\pospart{A} \negpart{A} = 0$.

We denote the maximally entangled state as $\ket{\Phi} = \sum_i \frac{1}{\sqrt{d}} \ket{ii}$.
A state $\rho$ is PPT if $\pt{\rho} \geq 0$, and we denote the collection of such states as $\PPT$.
A state $\rho$ is classically correlated if $\rho = \sum_{ij} p_{ij} \proj{a_i b_j}$ for some probabilities $p_{ij}$ and $\{\ket{a_i}\}, \{\ket{b_j}\}$ orthonormal bases of the local Hilbert spaces.
We denote the collection of classically correlated states as $\CC$.
We denote the collection of product states as $\PROD$.
The negativity of a state $\rho$ is defined as $N(\rho) = \half \left(\norm{\pt{\rho}} - 1\right)$.

\section{Partial transpose distance}

We first introduce the partial transpose distance and discuss its main properties.
Then we explore its connection to negativity.

\begin{definition}
	Let $\rho, \sigma \in \dm$ be two density matrices.
	The partial transpose distance $d_T(\rho, \sigma)$ is defined as
	\begin{align}
		d_T(\rho, \sigma)
		&= \half \norm{\pt{\rho} - \pt{\sigma}}.
	\end{align}
\end{definition}

Although the definition of partial transpose is basis-dependent, trace distance is invariant under unitaries.
Therefore the partial transpose distance does not depend on the choice of basis, or which subsystem we perform the transposition on.

Note that the partial transpose distance is closely related to trace distance, and inherits many of its properties.
For example, since the trace distance is a metric and partial transposition is a linear operation, the partial transpose distance is also a metric.
It admits an operational interpretation in the context of data hiding.
Recall the setting of Helstrom's theorem~\cite{helstrom1976quantum}: two states $\rho$ and $\sigma$ occur with equal probability $p = \half$, and the task is to unambiguously discriminate them.
This means that an agent must choose a positive operator-valued measure (POVM) $\{ M, \identity - M \}$ corresponding to a generalized measurement, such that the outcome $M$ occurs when the actual state is $\rho$, and the outcome $\identity - M$ occurs when the actual state is $\sigma$.
The goal is to choose the measurement to maximize the probability of success, shown by Helstrom's theorem to be
\begin{equation}
p_{\textrm{POVM}} \le \half \left(
		1 + \half \norm{\rho - \sigma}
	\right).
\end{equation}
In data hiding, we restrict the set of allowed POVMs~\cite{Terhal_2001,DiVincenzo_2002,Matthews_2009}.
For example, in PPT data hiding we are restricted to PPT POVMs, i.e.\ the POVM elements must be PPT operators~\cite{Rains_1999,Rains_1999a}.
By following the steps that lead to the Helstrom bound, but for PPT POVMs, we find
\begin{equation}
p_{\textrm{PPT}} \le \half \left(
		1 + d_T (\rho, \sigma)
	\right).
\end{equation}
Since the partial transpose distance between two states can be significantly greater than $1$, the bound above is in general not tight and may even be trivial.
However if the states that are chosen for the hiding scheme are PPT, such as the example in~\cite{Eggeling_2002},
then the partial transpose distance is bounded by $1$, and we always get a non-trivial bound.
As an upper bound to this task, we expect that the partial transpose distance is contractive under PPT operations.
Recall that PPT operations (also known as PPT-preserving maps as they map PPT states to PPT states) are defined as maps $\Lambda$, such that both $\Lambda$ and $T_B \circ \Lambda \circ T_B$ are completely positive~\cite{Rains_1999,Rains_1999a}.
These are the appropriate free operations in the theory of PPT entanglement, as they cannot map PPT states to non-PPT states.
The following lemma shows that the partial transpose distance is indeed contractive.
\begin{lemma}
	For any PPT operation $\Lambda$ and density matrices $\rho, \sigma$, we have
	\begin{align}
		d_T\left(\Lambda (\rho), \Lambda (\sigma) \right)
		&\leq
		d_T(\rho, \sigma).
	\end{align}\label{lemma:contractive}
\begin{proof}
	The map $\Lambda$ is a PPT operation, so by definition $T_B \circ \Lambda \circ T_B$ is a CPTP map.
	Using the contractivity of trace norm under positive, trace-preserving maps~\cite{RUSKAI_1994}, we have
	\begin{align*}
		d_T (\rho, \sigma)
		&=
		\half \norm{T_B (\rho - \sigma) }
		\\
		&\geq
		\half \norm{(T_B \circ \Lambda \circ T_B) \circ T_B (\rho - \sigma) }
		\\
		&=
		\half \norm{T_B (\Lambda(\rho) - \Lambda(\sigma))}
		\\
		&=
		d_T( \Lambda(\rho), \Lambda(\sigma) )
	\end{align*}
\end{proof}
\end{lemma}

\subsection{Relation to negativity}

In order to describe negativity in terms of the partial transpose distance, let us begin by recalling that negativity captures how much the partially transposed state fails to be a state:
\begin{equation}
\inf_{\sigma \in \dm} \frac{1}{2} \norm{\pt{\rho} - \sigma} = N(\rho).
\label{EQ_N}
\end{equation}
This can be seen by first noting that negativity is the lower bound to the trace distance
\begin{align*}
\inf_{\sigma \in \dm} \frac{1}{2} \norm{\pt{\rho} - \sigma}
& \ge \inf_{\sigma \in \dm} \frac{1}{2} \left( \norm{\pt{\rho}} - \norm{\sigma} \right) \\
& = \frac{1}{2} \left( \norm{\pt{\rho}} - 1 \right) \\
& = N(\rho),
\end{align*}
and then verifying that the following choice of $\sigma$ saturates the lower bound:
\begin{equation}
\sigma_N = \frac{{\pospart{\pt{\rho}}}}{\Tr{{\pospart{\pt{\rho}}}}}.
\label{EQ_SIGMA_N}
\end{equation}
Equation~\eqref{EQ_N} can be written in terms of the partial transpose distance by replacing optimization over the states $\sigma$ by optimization over matrices $A$ whose partial transpose gives the set of states:
\begin{equation}
\inf_{\sigma \in \dm} \frac{1}{2} \norm{\pt{\rho} - \sigma}
= \inf_{\pt{A} \in \dm} d_T(\rho,A).
\end{equation}
Since partial transpose operation is the inverse of itself, the set of matrices $A$ is given by the set of all partially transposed states:
\begin{equation}
\inf_{A \in PT} d_T(\rho,A) = N(\rho).
\end{equation}
Our main claim is that this infimum is also achieved when we replace optimization over all partially transposed states with optimization over a smaller set, composed only of states which admit positive partial transposition.
A related quantity was previously considered in Ref.~\cite{Wang_2020b}, where sandwiched Rényi relative entropy was used instead of trace distance.
Fig.~\ref{FIG_GEO} explains the underlying geometrical picture.
\begin{conjecture}
	Let $\rho$ be a density matrix.
	Then
	\begin{align*}
		\inf_{\sigma \in \PPT} d_T(\rho, \sigma) &= N(\rho).
	\end{align*}\label{conj:neg}
\end{conjecture}
We will now prove this conjecture for a broad class of states and provide the supporting numerical evidence in general.

\begin{figure}[!b]
\centering
\includegraphics[width=0.6\textwidth]{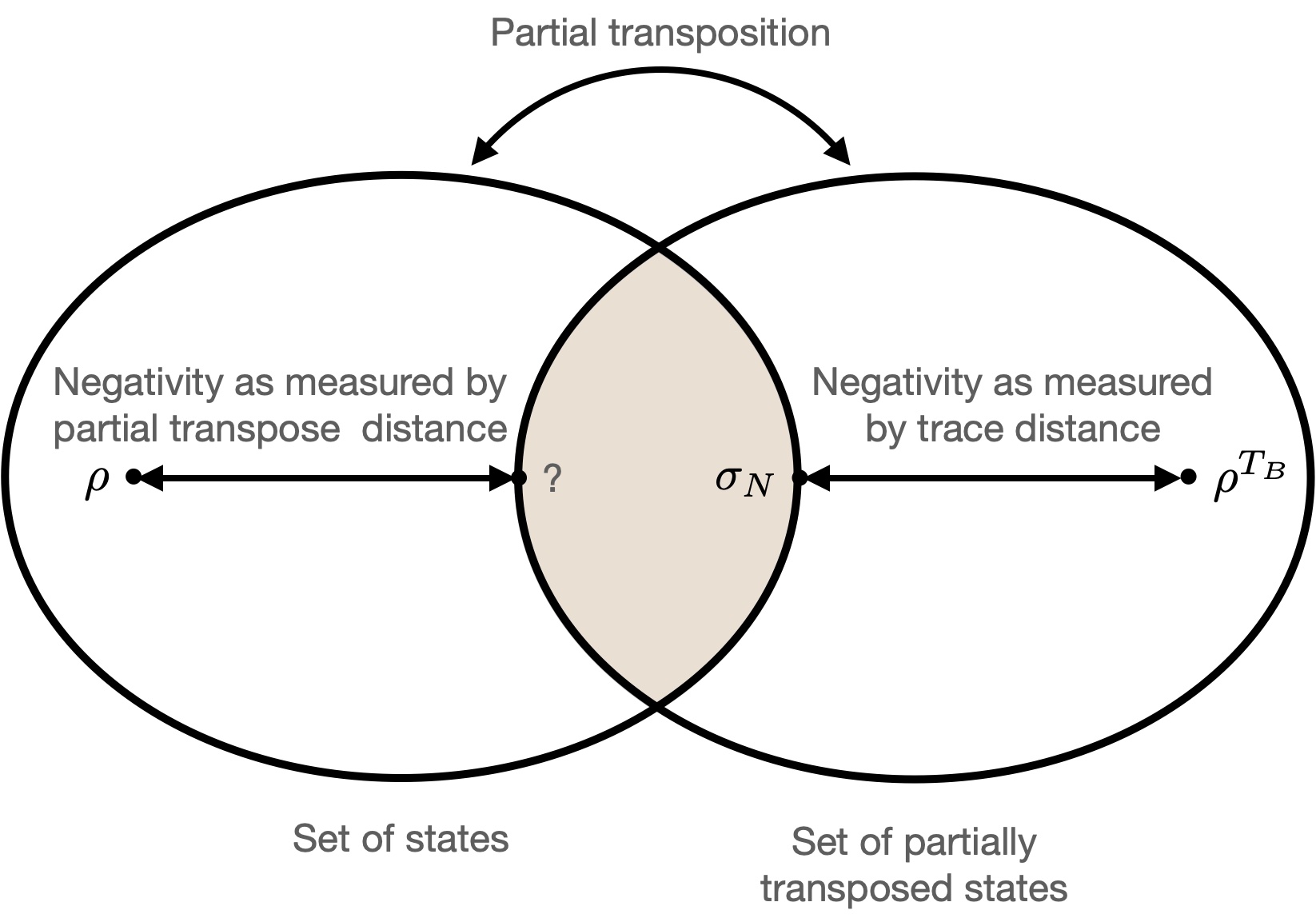}
\caption{\label{FIG_GEO}
Geometrical illustration. Negativity of state $\rho$ is the trace distance from its partial transposition to the set of states, see Eq.~\eqref{EQ_N}.
The closest state is known and given as $\sigma_N$ in the figure, see also Eq.~\eqref{EQ_SIGMA_N}.
We claim that alternatively, negativity can be computed by taking the partial transpose distance from $\rho$ to the set of PPT states (in the shadowed overlap of states and partially transposed states).
This is proven when partial transposition of $\sigma_N$ is positive semi-definite, i.e.\ a state.
In general, the closest PPT state in the partial transpose distance is unknown as depicted with the question mark.
We emphasise that this is just an illustration as partial transposition of $\sigma_N$ can produce matrices with negative eigenvalues, i.e. outside the set of states.
}
\end{figure}

\subsection{Evidence for the conjecture}

The argument with the triangle inequality given above applies just as well when the infimum is over the PPT states showing that negativity is a lower bound on the distance:
\begin{align}
	\inf_{\sigma \in \PPT} d_T(\rho, \sigma) &\ge N(\rho).
\end{align}
Fig.~\ref{FIG_GEO} suggests that a natural candidate for the closest state in the partial transpose distance is partial transposition of $\sigma_N$ introduced in Eq.~\eqref{EQ_SIGMA_N}.
Indeed, whenever $\sigma_N^{T_B}$ is a state the resulting partial transpose distance is given by the negativity as it reduces to Eq.~\eqref{EQ_N}.
This has effectively been shown for the set of Schmidt-correlated states, i.e.\ states of the form $\sum_{ij} a_{ij} \ket{ii} \bra{jj}$, in Ref.~\cite{Khasin_2007}.
Now we note that all the states with positive binegativity~\cite{Audenaert_2002}, i.e.\ such that $\pt{|\pt{\rho}|} \ge 0$, also admit $\sigma_N^{T_B}$ which is positive semi-definite.
Indeed:
\begin{equation}
\pt{\pospart{\pt{\rho}}} = \frac{1}{2} \pt{\left( |\pt{\rho}| + \pt{\rho} \right)} \geq 0.
\end{equation}

Binegativity was introduced as a mathematical tool to bound PPT exact entanglement cost by logarithmic negativity~\cite{Audenaert_2003}.
It is known that all pure states, two-mode Gaussian states~\cite{Audenaert_2002}, and two-qubit states~\cite{Ishizaka_2004} have positive binegativity.
For this broad class of states, the introduced distance is therefore equal to the negativity.
Recall that one of the interpretations of negativity is as a bound on PPT exact entanglement cost~\cite{Audenaert_2003}---for states with positive binegativity they are exactly equal.
Thus negativity has an exact operational interpretation for these states.
It is interesting to note that for these states, negativity is equal to the partial transpose distance to PPT states.

In general, we do not know a closed expression for the partial transpose distance to PPT states.
However, the optimization can be expressed as a semidefinite program (SDP)~\cite{Vandenberghe_1996}.
Recall that the computation of trace norm can be formulated as an SDP~\cite{Watrous_2018}.
Since the set of PPT states is also characterized by linear positivity conditions, the whole optimization can be phrased as an SDP instance and allows us to solve the problem efficiently on a computer.
We have sampled uniformly mixed states of two qudits, according to measure induced by partial tracing~\cite{Zyczkowski_2001,Braunstein_1996,Hall_1998}, for dimensions $d = 2$ to $d=6$.
For each $d$, we sampled $10^6$ mixed states and used the SDP formulation to find the closest distance to the set of PPT states.
In all cases, we have never observed a counterexample to Conjecture~\ref{conj:neg}.

The situation is even more interesting because there exist states for which $\sigma_N^{T_B}$ is not positive semi-definite
and yet in these cases numerical analysis confirms the conjecture.
A simple example of such a state is given by even mixture of pure states $\ket{00} + \ket{01} + \ket{12}$ and $\ket{10} + \ket{21} + \ket{22}$.
In principle, the lack of counterexamples could be attributed to two numerical issues:
the negative eigenvalues of $\pt{\pospart{\pt{\rho}}}$ are very small or the set of counterexamples might be of zero measure.
In any case, the numerical analysis shows that negativity is guaranteed to be a good lower bound for a generic state.

Coming back to the problem of computing the PPT exact entanglement cost, recently a closed form of the quantity was obtained~\cite{Wang_2020}.
It was shown that the PPT exact entanglement cost is given by the $\kappa$-entanglement.
Furthermore, examples of states with $\kappa$-entanglement strictly larger than logarithmic negativity were given.
For these examples, numerical results still confirm that the conjecture is true.

Finally, in Appendix~\ref{app:conjecture-alt} we obtain an alternative form of the conjecture --- Conjecture~\ref{conj:neg} is true if and only if for every state $\rho$ there exists a PPT state $\sigma$ such that $\pospart{\pt{\rho}} \geq \sigma$.
This alternative expression provides a more efficient SDP formulation of the conjecture.
It also provides a partial answer to the following problem discussed in Refs.~\cite{Miranowicz_2008a,Friedland_2011,Girard_2014}: find a set of entangled states that admit a given closest separable state that we now extend to the closest PPT state.
For a PPT state $\sigma$, all entangled states $\rho$ that satisfy $\pospart{\pt{\rho}} \geq \pt{\sigma}$ have $\sigma$ as their closest PPT state.

\section{Hierarchy of correlation quantifiers}

We have shown that negativity is a lower bound on the partial transpose distance to the set of PPT states.
In addition, we provided evidence to the conjecture that the two quantities are equal for all states.
Furthermore, there exists an SDP to compute the distance to PPT states.
Thus we have an easily computable quantifier of PPT entanglement.

We now define the correlation quantifiers on equal footing (measured by the same distance) to PPT entanglement.
\begin{definition}
	Let $\rho$ be a quantum state.
	We define the quantifier of correlation $X$ as
	\begin{align*}
		Q_{X} (\rho)
		&= \inf_{\sigma \in X} d_T \left( \rho, \sigma \right),
	\end{align*}
	where $X$ is the set of uncorrelated states.
	Particular examples are PPT entanglement ($X = \PPT$), non-classical correlation ($X = \CC$), and total correlation ($X = \PROD$).
\end{definition}

Since $d_T\left( \rho, \sigma \right)$ vanishes if and only if $\rho = \sigma$, these measures vanish only on uncorrelated states.
Furthermore, because the partial transpose distance is contractive under PPT operations, it is straightforward to show that these quantities are monotonic under the subset of PPT operations that preserves the set of uncorrelated states (also known as free operations in resource theory parallel).
Since the sets $\PROD \subset \CC \subset \PPT$ form a hierarchy, the corresponding quantifiers are also ordered and comparable.

If the set of uncorrelated states is convex, then the computation of $Q_{X}$ is a convex optimization problem.
We can see this by noticing that the partial transpose distance is jointly convex.
From the general properties of convex problems, it follows that any local minima must be global minima too.
Unfortunately neither the set of product states nor the set of classical states are convex and therefore it is an open problem whether the corresponding distances are efficiently computable.
In any case, one expects this difficulty to be less severe than the computation of entanglement using other distances,
because of the smaller number of parameters needed to parametrize product states compared to separable states.

In what follows, we find closed form expressions of partial transpose distances for pure states and a class of classical states.
These examples show intriguing structure of correlations as measured by the partial transpose distance which is captured by a particular subadditivity relation.
Whenever possible, we will compare them with quantifiers obtained using relative entropy distance.

\subsection{Pure states}

Recall that for pure states, the closest separable state in relative entropy is obtained by performing a local dephasing on the Schmidt basis~\cite{Vedral_1998}.
Thus the closest separable state is also the closest classically-correlated state, and therefore the distance to separable states is equal to the distance to classically-correlated states.
In other words, the only form of quantum correlations in pure states is given by entanglement.
While the closest separable state in partial transpose distance is no longer unique, the statement that the only form of quantum correlation in pure states is given by entanglement remains also for the partial transpose quantifiers.

\begin{theorem}
	Let $\psi$ be a pure quantum state.
	Then
	\begin{align*}
		Q_{PPT} (\psi) = Q_{CC} (\psi) = N(\psi).
	\end{align*}\label{th:pure-cc}
\end{theorem}
\begin{proof}
	Recall that pure states have positive binegativity~\cite{Audenaert_2002} and therefore $N(\psi) = Q_{PPT} (\psi) \leq Q_{CC} (\psi)$.
	Thus, we only need to show that there is a classically correlated state $\sigma$ such that $d_T \left( \psi, \sigma \right) = N(\psi)$.
	To this end, we write $\psi$ in its Schmidt basis $\ket{\psi} = \sum_i \sqrt{p_i} \ket{ii}$, and define $\sigma = \sum_i p_i \proj{ii}$.
	Then
	\begin{align*}
		d_T \left( \psi, \sigma \right)
		&= \half \norm{\pt{\psi} - \pt{\sigma} }
		\\
		&= \half \norm{\sum_{ij} \sqrt{p_i p_j} \ket{ij} \bra{ji}  - \sum_i p_i \proj{ii} }
		\\
		&= \half \norm{\sum_{i \neq j} \sqrt{p_i p_j} \ket{ij} \bra{ji} }.
	\end{align*}
	Notice that the matrix $\sum_{i \neq j} \sqrt{p_i p_j} \ket{ij} \bra{ji}$ has eigenvalues $\pm \sqrt{p_i p_j}$ with eigenvectors $\ket{ij} \pm \ket{ji}$, for $i < j$.
	Thus $d_T \left( \psi, \sigma \right) = \sum_{i < j} \sqrt{p_i p_j} = N(\psi)$.
\end{proof}

This similarity to relative entropy distance should not be extrapolated as shown in the next section on example of classically-correlated states.

\subsection{Classically-correlated states}

We wish to compare to the relative entropy distance again.
It was shown that in relative entropy, the closest product state is always the product of marginals~\cite{Modi_2010} and hence the distance is given by the quantum mutual information.
However this is no longer true for trace distance~\cite{Aaronson_2013} and also does not hold for partial transpose distance as we now show using a class of classically-correlated states.

\begin{theorem}
	Let $\sigma = \sum_i p_i \proj{ii}$ be a classically-correlated state, with $p_i$ in non-increasing order.
	Let $m
		= \max{\left\{
				n\,|\, \sum_{i < n} \sqrt{p_i} \leq 1
		\right\}}
		$.
	Then
	\begin{align}
		Q_{Prod} (\sigma)
		&=
		1 - \sum_{i < m} p_i - {\left(1 - \sum_{i < m} \sqrt{p_i}\right)}^2
		\label{eq:cc-prod}
	\end{align}
	and the closest product state is given by:
	\begin{equation}
	\chi = {\left( \sum_{i < m} \sqrt{p_i} \proj{i} + \left(1 - \sum_{i < m} \sqrt{p_i} \right) \proj{m} \right)}^{\otimes 2}
	\end{equation}\label{th:cc-product}
\end{theorem}
\begin{proof}
	Proof is given in Appendix~\ref{app:cc-prod}.
\end{proof}

We now have the tools to prove an unexpected relation connecting the classical and quantum correlations in a pure state as measured by the partial transpose distance.

\subsection{Subadditivity of classical and quantum correlations in pure states}

We have shown that the only form of quantum correlations in pure states is given by entanglement.
If the distance in state space is chosen to be relative entropy, it turns out that entanglement and classical correlations sum up to exactly the total correlations in a pure state.
More formally, the relative entropy of entanglement $E_R(\psi) = \inf_{\sigma \in \textrm{SEP}} S(\psi \| \sigma)$ and classical correlations defined by the distance of the closest classically-correlated state $\sigma_0$ to the set of product states $C_R(\psi) = \inf_{\chi \in \PROD} S(\sigma_0 \| \chi)$ satisfy
\begin{align*}
	I(\psi) = E_R(\psi) + C_R(\psi),
\end{align*}
where the mutual information $I(\psi) = \inf_{\chi \in \PROD} S(\psi \| \chi)$ is the distance from the original state to the set of product states.
Since one of the closest classical states in partial transpose distance is the same as in relative entropy, we might hope that the same additivity relation holds for partial transpose distance.
A simple application of triangle inequality shows that $I_T(\psi) \le N(\psi) + C_T(\psi)$,
where we used analogical definitions as for the relative entropy, i.e. $I_T(\psi) = \inf_{\chi \in \PROD} d_T(\psi \| \chi)$ and $C_T(\psi) = \inf_{\chi \in \PROD} d_T(\sigma_0 \| \chi)$.
It turns out that this inequality is always strict and the following stronger inequality holds, where the classical correlations are multiplied by a factor of $\frac{1}{2}$.
\begin{theorem}\label{th:pure-decomposition}
	Let $\psi$ be a pure state.
	Then
	\begin{align}
		I_T(\psi)
		&\leq
		N(\psi) + \half C_T(\psi).
	\label{EQ_ADD_TRANS}
	\end{align}
	If $\psi$ is maximally entangled, then we have equality.
\end{theorem}
\begin{proof}
	Proof is given in Appendix~\ref{app:pure-decomposition}.
\end{proof}

Notice that the diameter (the largest distance between two states in a set) of the set of PPT states is bounded by $1$ for any dimension $\dim(\hs_A \otimes \hs_B)$,
whereas the diameter of $\dm$ grows linearly as $\min\{\dim(\hs_A), \dim(\hs_B)\}$.
Thus, as the dimension grows, the classical correlations $C_T(\psi)$ contribute less and less to the total correlations of a pure state $I_T(\psi)$.

We also performed numerical checks ($10^3$ uniformly random pure states in dimensions $d=2$ to $d=6$) which suggest that the equality in Eq.~\eqref{EQ_ADD_TRANS} is achieved for all pure states.

\section{Applications}
We show two applications of the developed framework: to detect whether a dynamics is decomposable and to detect non-Markovianity.

\subsection{Detection of non-decomposability}
Consider a tripartite system $\rho_{ABC}$ that evolves in time according to some global map $\Lambda_{ABC}$.
A decomposition of the tripartite map $\Lambda_{ABC}$ are maps $\Lambda_{AC}, \Lambda_{BC}$ such that $\Lambda_{ABC} = \Lambda_{AC} \Lambda_{BC}$.
If a map has a decomposition, then we say that the map is decomposable.
Intuitively, decomposable maps are those that can be implemented by having a mediator particle $C$ and the following scheme: $A$ interacts with $C$, sends the mediator $C$ to $B$, and then $B$ interacts with $C$.

Ref.~\cite{Krisnanda_2018} shows the gain in correlations between $A$ and $B$ under decomposable dynamics is bounded above.
Thus a gain above this bound detects the non-decomposability.
For correlation quantifiers within the distance approach, the bound holds for any initial state.
However for a general correlation quantifier, only product initial states were allowed.
Hence, only such states could be considered with negativity up to now.
The approach developed here allows us to derive a bound for negativity that holds for an arbitrary initial state.

In order to show the usefulness of this extension, let us give an example where the non-decomposability is detected with negativity and the dynamics starts in a non-product tripartite state.
Consider two fields coupled by a two-level atom and described by the Jaynes-Cummings Hamiltonian:
\begin{equation}
H = g(\hat a \sigma_+ + \hat a^\dagger \sigma_-) + g(\hat b \sigma_+ + \hat b^\dagger \sigma_-),
\end{equation}
where $\hat a$ and $\hat b$ are the annihilation operators for the fields $A$ and $B$, respectively, and $\sigma_{\pm}$ are the raising and lowering operators for the atom, which we will denote as object $C$.
Theorem 1 in Ref.~\cite{Krisnanda_2018} applied to the partial transpose distance shows the following bound which holds for decomposable dynamics independently of the initial state:
\begin{equation}
	N_{A:B} (t) \le I_T^{AC:B}(0) + \sup_{\ket{\psi_{AC}}} N_{A:C} (\psi_{AC}),
\end{equation}
where $I_T^{AC:B}(0) = \inf_{\sigma_{AC} \otimes \sigma_B} d_{T_B}(\rho_0, \sigma_{AC} \otimes \sigma_B)$ is the partial transpose distance from the initial state $\rho_0$
to the set of product states in partition $AC:B$ and the supremum is over all pure states in $A:C$ bipartition.
One verifies that the latter is given by $\half \left(d_C - 1 \right)$, which in our case is $\half$.
Assuming that the initial state is pure, using \eqref{EQ_ADD_TRANS} and $C_T(\psi) \leq 1$, we find that any decomposable dynamics satisfies:
\begin{equation}
	N_{A:B} (t) \le N_{AC:B}(0) + 1.
	\label{EQ_NN1}
\end{equation}
For an illustration of negativity witnessing non-decomposability with a non-product initial state, consider the NOON state on $AB$ and ground state of $C$, $|\psi_{ABC} \rangle = \frac{1}{\sqrt{2}}(|400 \rangle + |040 \rangle)$.
In this case $N_{AC:B}(0) = \frac{1}{2}$, which transforms Eq.~\eqref{EQ_NN1} to $N_{A:B} (t) \leq \frac{3}{2}$.
Fig.~\ref{FIG_NONDIV} shows that this bound is indeed violated, witnessing the non-decomposability.
Note that in this method, the interaction Hamiltonians remain unknown, i.e.\ the bound does not depend on them.
\begin{figure}[!b]
\center
\includegraphics[width=0.8\textwidth]{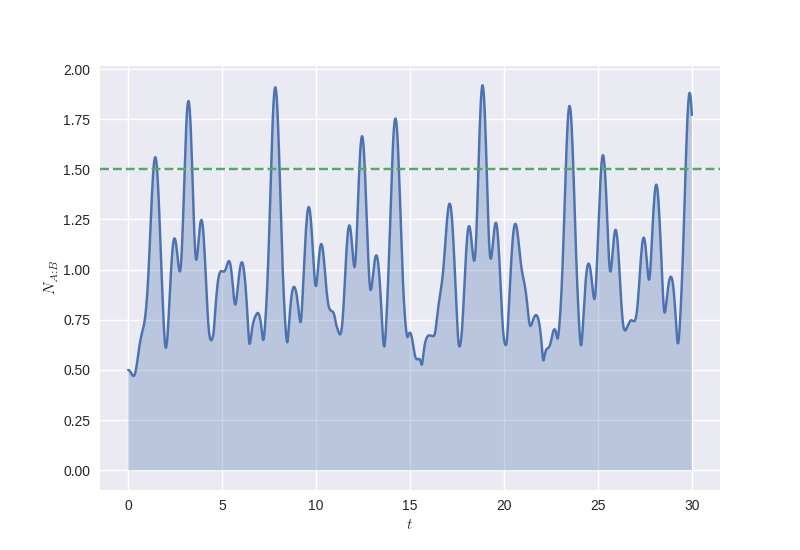}
\caption{\label{FIG_NONDIV}
Negativity of two fields coupled via two-level atom and evolving under Jaynes-Cummings model.
The dashed line is the bound on negativity in any divisible dynamics derived using the fact that negativity is the lower bound on the partial transpose distance to the set of PPT states.
The bound is useful because it is violated, revealing non-decomposability of the underlying dynamics.}
\end{figure}

\subsection{Detection of non-Markovianity}
The problem of witnessing non-Markovian dynamics has been a fruitful area of study~\cite{Rivas_2014}.
We stick here to the definition of Markovian dynamics in terms of dynamical maps.
A dynamical map is a continuous, time-parametrized collection of CPTP maps $\{ \Lambda_t \}$ that describes the evolution of a system.
We require that $\Lambda_0 = \identity$.
Such a dynamical map is called divisible if for any time sequence $0 < s < t$, there exists a ``propagator'' $V_{t,s}$ such that $\Lambda_t = V_{t, s} \Lambda_s$.
A dynamical map is called CP-divisible if it is divisible and the propagator is a CPTP map.
Such dynamics is shortly called Markovian.

A well-known witness of non-Markovianity is the increase of the trace distance between two states in the evolution~\cite{Breuer_2009}.
Ref.~\cite{Bylicka_2017} showed this witness is in fact universal, i.e.\ if a given dynamics is non-Markovian, we can always find two states whose trace distance increases (after adding ancillas).
The same methods applied to the partial transpose distance show that the dynamical map $\{ \Lambda_t \}$ is Markovian (CP-divisible) if and only if:
\begin{align}\label{EQ_MARK}
	\frac{d}{dt} d_{T_C} \left(
		(\Lambda_t \otimes \identity_{BC}) \rho,
		(\Lambda_t \otimes \identity_{BC}) \sigma
	\right) \leq 0 \quad \textrm{for all} \quad \rho, \sigma \in \mathcal{D} (\hs_A \otimes \hs_B \otimes \hs_C).
\end{align}
With the developed tools we can show that negativity change is a witness of non-Markovianity.
Let us denote the evolved states $\rho_t = (\Lambda_t \otimes \identity_{BC}) \rho$ and $\sigma_t = (\Lambda_t \otimes \identity_{BC}) \sigma$, and choose $\sigma$ as the closest PPT state to $\rho$ in AB:C partition (for the partial transpose distance). Then also $\sigma_t$ is PPT and we have the following chain of inequalities for CP-divisible dynamics:
\begin{equation}
N_{AB:C}(\rho_t) \le d_{T_C}(\rho_t,\sigma_t) \le d_{T_C}(\rho,\sigma) = N_{AB:C}(\rho).
\end{equation}
Accordingly, if the negativity increases one concludes that the map was not Markovian.
Ref.~\cite{Kolodynski_2020} showed that negativity is a universal witness of non-Markovianity, i.e.\ for any non-Markovian map there exist states for which the negativity increases.
Furthermore, since this proof relies only on the monotonicity of the partial transpose distance, similar witness can be formulated for e.g.\ total correlations $I_T$ introduced here.

\section{Conclusions}

Quantum mechanical correlations come in many flavours and ideally we would like to have a framework in which they all could be efficiently computed and meaningfully compared.
In this spirit, we introduced partial transpose distance and explored its connection to negativity --- a computable entanglement monotone.
For a broad class of states we proved that negativity is equal to the partial transpose distance to the set of PPT states and we conjecture this relation holds in general.
We then defined other types of correlations on equal footing to negativity (measured by the same distance) and derived their closed forms for selected classes of states.
These findings allowed us to show a subadditivity relation where total correlations according to the partial transpose distance are less than quantum correlations summed with one half of classical correlations.
Finally, we applied developed framework to witnessing divisibility properties of quantum dynamics.
The paper contains a number of open problems that we hope will stimulate further development of comparable and computable correlation quantifiers.

\section*{Acknowledgements}

RG thanks Mark Wilde for discussions.
This work is supported by Foundation for Polish Science (FNP), IRAP project ICTQT, contract no. 2018/MAB/5, co-financed by EU Smart Growth Operational Programme.
RG acknowledges support from the UGrants-start awarded by the University of Gda\'nsk.
MM acknowledges financial support by the ``Quantum Coherence and
Entanglement for Quantum Technology'' project, carried out within the
First Team programme of the Foundation for Polish Science co-financed by
the European Union under the European Regional Development Fund.
TP is supported by the Polish National Agency for Academic Exchange NAWA Project No. PPN/PPO/2018/1/00007/U/00001.

\appendix

\section{Robustness of total correlation}\label{app:robustness}
Given a state $\rho$, we define (entanglement) robustness of $\rho$ relative to $\sigma$ as~\cite{Vidal_1999a}
\begin{align*}
	R_{\mathrm{SEP}}(\rho \| \sigma)
	&=
	\inf \left\{
		s
		\,|\, s \geq 0, \frac{1}{1+s} \left( \rho + s \sigma \right) \in \SEP
	\right\},
\end{align*}
where $\SEP$ denotes the set of separable states.
Vidal~\cite{Vidal_1999a} defined the absolute robustness of the state $\rho$ as $R_{\SEP} (\rho) = \inf_{\sigma \in \SEP} R_{\mathrm{SEP}}(\rho \| \sigma)$, and it was shown that this quantity fulfills the axioms of an entanglement measure.

By swapping the set of separable states $\SEP$ with product states $\PROD$ in the definition of $R_{\mathrm{SEP}}(\rho \| \sigma)$ and $R_{\SEP} (\rho)$, we get a candidate quantifier of total correlation
\begin{align}
	R_{\mathrm{Prod}}(\rho \| \sigma)
	&=
	\inf \left\{
		s
		\,|\, s \geq 0, \frac{1}{1+s} \left( \rho + s \sigma \right) \in \PROD
	\right\},
	\nonumber \\
	R_{\PROD} (\rho)
	&=
	\inf_{\sigma \in \PROD} R_{\mathrm{Prod}}(\rho \| \sigma).
\end{align}
However, this measure cannot distinguish between entangled states, because $R_{\PROD} (\psi) = \infty$ for any entangled pure state $\psi$.
This is a consequence of the following lemma.
\begin{lemma}
	Let $\psi$ be an entangled pure state.
	Then for any $\chi_A, \chi_B \geq 0$, $p > 0$,
	\begin{align*}
		\rho
		&=
		p \proj{\psi} + (1-p) \chi_A \otimes \chi_B
	\end{align*}
	is not a product state.
\end{lemma}
\begin{proof}
Clearly if $p = 1$, $\rho$ cannot be product because $\rho = \psi$ is entangled.
So without loss of generality, we assume that $0 < p < 1$.

We prove by contradiction.
Assume $\rho$ can be written as a product $\rho = \rho_A \otimes \rho_B$.
Then $\rho_A = p \alpha + (1-p) \chi_A$ and $\rho_B = p \beta + (1-p) \chi_B$, where $\alpha$ and $\beta$ are the marginals of $\psi$.
Since $\psi$ is entangled, it has at least two terms in the Schmidt decomposition $\ket{\psi} = r_0 \ket{00} + r_1 \ket{11} + \dots$.
Note that in this basis, the reduced states $\alpha$ and $\beta$ are fully diagonal.
We compute the matrix element $\langle 01 | \rho |10 \rangle$ in two ways: first using $\rho = \rho_A \otimes \rho_B$ and second using the expansion in terms of the entangled state.
One finds
\begin{equation}
	{(1-p)}^2 \bra{0} \chi_A \ket{1} \bra{1} \chi_B \ket{0}
	=
	\bra{01} \rho \ket{10}
	=
	(1-p) \bra{0} \chi_A \ket{1} \bra{1} \chi_B \ket{0}
\end{equation}
If both local states $\chi_A$ and $\chi_B$ admit off-diagonal elements, this implies either $p = 0$ or $p = 1$, which contradicts our assumptions.
Therefore, at least one of them has to be diagonal in the Schmidt basis.
This, however, contradicts matrix element $\langle 00 | \rho |11 \rangle$ which in such a case implies that $p \, r_0 r_1 = 0$.
Therefore for any $p > 0$, $\rho$ is not a product.
\end{proof}

\section{Alternative form of conjecture}\label{app:conjecture-alt}
We will show that the conjecture holds if and only if for every $\rho$, one can find a PPT state $\sigma$ such that $\pospart{\pt{\rho}} \geq \sigma$.
It is easy to show if there is a PPT state $\sigma$ such that $\pospart{\pt{\rho}} \geq \sigma$, then $d_T(\rho, \pt{\sigma}) = N(\rho)$.
Therefore, we only have to show the conjecture implies there exists $\sigma \in \PPT$ such that $\pospart{\pt{\rho}} \geq \sigma$.

Suppose the conjecture holds, i.e.\ there exists $\sigma \in \PPT$ such that $d_T (\rho, \sigma) = N(\rho)$.
Let us choose the optimal $\sigma$ and show that this implies $\pospart{\pt{\rho}} \geq \pt{\sigma}$.
We have
\begin{align}
	\norm{\pt{\rho} - \pt{\sigma}}
	=
	2 d_T (\rho, \sigma)
	= 2 N(\rho)
	= \norm{\pt{\rho}} - \norm{\pt{\sigma}}.
\end{align}

Let us write $\pt{\sigma}$ in the blocks corresponding to $\pospart{\pt{\rho}}$ and $\negpart{\pt{\rho}}$, that is $\pt{\sigma} = \begin{pmatrix} X & Z \\ Z^\dagger & Y \end{pmatrix}$, with $X, Y \geq 0$.
By projecting onto the diagonal blocks, we have
\begin{align}
	\norm{\pt{\rho} - \pt{\sigma}}
	&=
	\norm{
		\begin{pmatrix}
			\pospart{\pt{\rho}} & 0 \\
			0 & \negpart{\pt{\rho}}
		\end{pmatrix}
		- \begin{pmatrix}
			X & Z \\
			Z^\dagger & Y
		\end{pmatrix}
	}
	\nonumber \\
	&\geq
	\norm{
		\begin{pmatrix}
			\pospart{\pt{\rho}} & 0 \\
			0 & \negpart{\pt{\rho}}
		\end{pmatrix}
		- \begin{pmatrix}
			X & 0 \\
			0 & Y
		\end{pmatrix}
	}
	\nonumber \\
	&=
	\norm{\pospart{\pt{\rho}} - X}
	+ \norm{\negpart{\pt{\rho}} - Y},
\end{align}
where we used the contractivity of trace norm in the inequality.

We have reduced the problem to showing that $\norm{\pospart{\pt{\rho}} - X} + \norm{\negpart{\pt{\rho}} - Y} = \norm{\pt{\rho}} - \norm{\pt{\sigma}}$ implies $\pospart{\pt{\rho}} \geq \pt{\sigma}$.
Using the inequality $\norm{A} \geq \Tr{A}$ for the first term and triangle inequality for the second, we have
\begin{align*}
	\norm{\pospart{\pt{\rho}} - X}
	+ \norm{\negpart{\pt{\rho}} - Y}
	&\geq
	\Tr{\left( \pospart{\pt{\rho}} - X \right)}
	+ \norm{\negpart{\pt{\rho}}}
	- \norm{Y}
	\\
	&=
	\norm{\pt{\rho}}
	- \norm{\pt{\sigma}}
\end{align*}
Because we chose $\sigma$ such that $\norm{\pt{\rho} - \pt{\sigma}} = \norm{\pt{\rho}} - \norm{\pt{\sigma}}$, the inequalities applied to the first and second term must be saturated.
Therefore, the following conditions must be satisfied:
\begin{enumerate}
	\item\label{it:cond-1} $\norm{\pospart{\pt{\rho}} - X} = \Tr{\pospart{\pt{\rho}} - X}$, and
	\item\label{it:cond-2} $\norm{\negpart{\pt{\rho}} - Y} = \norm{\negpart{\pt{\rho}}} - \norm{Y}$.
\end{enumerate}
Since $\negpart{\pt{\rho}} - Y \leq 0$, condition~\ref{it:cond-2} implies
\begin{align}
	-\Tr{\left(\negpart{\pt{\rho}} - Y\right)} = \norm{\negpart{\pt{\rho}} - Y} = \norm{\negpart{\pt{\rho}}} - \norm{Y} = - \Tr{\negpart{\pt{\rho}}} - \Tr{Y},
\end{align}
i.e.\ $\Tr{Y} = 0$.
However since $Y \geq 0$, this can only be true if $Y = 0$.
Since $Y = 0$ and $\pt{\sigma} = \begin{pmatrix} X & Z \\ Z^\dagger & 0 \end{pmatrix} \geq 0$, we must have $Z = 0$ and $\pt{\sigma} = X$.
Therefore condition~\ref{it:cond-1} becomes $\norm{\pospart{\pt{\rho}} - \pt{\sigma}} = \Tr{\pospart{\pt{\rho}} - \pt{\sigma}}$, which is true if and only if $\pospart{\pt{\rho}} \geq \pt{\sigma}$.

\section{Proof of Theorem~\ref{th:cc-product}}\label{app:cc-prod}
\begin{theorem*}
	Let $\sigma = \sum_i p_i \proj{ii}$ be a classically-correlated state, with $p_i$ in non-increasing order.
	Let $m
		= \max{\left\{
				n\,|\, \sum_{i < n} \sqrt{p_i} \leq 1
		\right\}}
		$.
	Then
	\begin{align*}
		Q_{Prod} (\sigma)
		&=
		1 - \sum_{i < m} p_i - {\left(1 - \sum_{i < m} \sqrt{p_i}\right)}^2
	\end{align*}
	and the closest product state is given by:
	\begin{align*}
	\chi = {\left( \sum_{i < m} \sqrt{p_i} \proj{i} + \left(1 - \sum_{i < m} \sqrt{p_i} \right) \proj{m} \right)}^{\otimes 2}
	\end{align*}
\end{theorem*}
\begin{proof}
	Note that the state $\sigma$ is diagonal in the basis $\ket{ij}$, so it is invariant under the projection $A \mapsto \sum_{ij} \proj{ij} A \proj{ij}$.
	Since the projection takes product states to product states, a standard argument using contractivity of the distance shows that the closest product state to $\sigma$ must also be diagonal in the basis $\ket{ij}$:
	\begin{align*}
		\chi &= \left(
			\sum_i r_i \proj{i}
			\right) \otimes \left(
			\sum_j s_j \proj{j}
		\right),
	\end{align*}
	with $r_i, s_j \geq 0$, and $\sum_i r_i = \sum_j s_j = 1$.

	We will optimize over $\chi$ and show that the minimum is given by Eq.~\eqref{eq:cc-prod}.
	We have
	\begin{align*}
		\norm{\sigma - \chi}
		&= \norm{\sum_i p_i \proj{ii}
			- \sum_{ij} r_i s_j \proj{ij}
		}
		\\
		&= \sum_{ij} |p_i \delta_{ij} - r_i s_j|.
	\end{align*}
	Separating out the $i = j$ terms from $i \neq j$, we have
	\begin{align*}
		\norm{\sigma - \chi}
		&=
		\sum_i |p_i - r_i s_i|
		+ \sum_{ij} \left(1 - \delta_{ij} \right) r_i s_j
		\\
		&=
		\sum_i |p_i - r_i s_i|
		+ \sum_{i} r_i (1 - s_i)
		\\
		&=
		\sum_i |p_i - r_i s_i|
		+ 1
		- \sum_{i} r_i s_i
		\\
		&=
		\sum_i |p_i - r_i s_i|
		+ \sum_i (p_i - r_i s_i).
	\end{align*}
	Using the identity $|x| + x = \max\{0,\, 2x\}$, we get
	\begin{align*}
		\norm{\sigma - \chi}
		&= \sum_i \max \{0,\, 2(p_i - r_i s_i) \}.
	\end{align*}
	Writing $u_i = \frac{r_i + s_i}{2}, v_i = \frac{r_i - s_i}{2}$, we have
	\begin{align*}
		\norm{\sigma - \chi}
		&= \sum_i \max \{0,\, 2(p_i - (u_i + v_i)(u_i - v_i)) \}
		\\
		&= \sum_i \max \{0,\, 2(p_i - u_i^2 + v_i^2) \}.
	\end{align*}
	We now find conditions under which this distance is minimized.
	Minimization over $v_i$'s gives all $v_i = 0$, i.e.\ the local states are the same $r_i = s_i = u_i$.
	Hence
	\begin{align}
		Q_{Prod} (\sigma)
		&= \inf_{r_i, s_j} \half \norm{\sigma - \chi}
		\nonumber
		\\
		&= \inf_{u_i} \sum_i \max \{0,\, p_i - u_i^2\}.
		\label{eq:traceDist}
	\end{align}

	We will show that the minimum over $u_i$ is achieved when $\tilde{u}_i = \sqrt{p_i}$ for $i < m$, $\tilde{u}_m =  1 - \sum_{i < m} \sqrt{p_i}$, and zero otherwise.
	Indeed, suppose at first that for some index $i_0$, we have $u_{i_0} > \sqrt{p_{i_0}}$.
	We could then consider another state $\chi'$ given by the coefficients $u'_i$, such that $u'_{i_0} = \sqrt{p_{i_0}}$,
	$u'_{j_0} = u_{j_0 } + u_{i_0} - \sqrt{p_{i_0}}$,
	where $j_0$ is an arbitrary index not equal to $i_0$,
	and $u'_{i} = u_i$ whenever $i_0  \neq i \neq j_0$.
	Then clearly $ \sum_i u'_i = 1$,
	and $\sum_i \max \{0,\, p_i - u_i'^2\} < \sum_i \max \{0,\, p_i - u_i^2\}$.  Hence, for a state that minimizes Eq.~\eqref{eq:traceDist}, we must have $u_i \leq \sqrt{p_i}$ for all $i$.
	The problem of finding the optimal set of coefficients $u_i$ in Eq.~\eqref{eq:traceDist} reduces to maximization of the expression $\sum_i u_i^2$ under the constraints: $0 \leq u_i \leq \sqrt{p_i}$, for all $i$, and $\sum_i u_i = 1$.

	Suppose now that $u_1 < \sqrt{p_1}$.
	We can assume without loss of generality that $u_1>u_2$ as otherwise we would swap the two numbers and the new set of coefficients would still satisfy the constraints because $p_i$'s are assumed to be in non-increasing order.
	Consider a new set $u'_i$, defined as
	$u'_1 = u_1 + \epsilon$, $u'_{2} = u_{2} - \epsilon$,
	where $0 < \epsilon < \sqrt{p_1} - u_1$,
	and $u'_i = u_i$ whenever $i > 2$ (if $u'_2$ is negative we set $u'_2 = 0$ and choose $u'_3 = u_3 - (\epsilon - u_2)$ and so on).
	With this choice $\sum_i {(u'_i)}^2 > \sum_i u_i^2$.
	Therefore, in the maximal case we must have $\tilde{u}_1 = \sqrt{p_1}$.
	By repeating the same argument, $\tilde{u}_i = \sqrt{p_i}$ for all $i<m$.
	It is now easy to see that the maximum for the sum of the remaining coefficients, $\sum_{i\geq m} u_i^2$,
	is achieved for  $\tilde{u}_m = 1 - \sum_{i<m} \sqrt{p_i}$ and
	$\tilde{u}_i = 0$, for $i>m$.

	Thus, Eq.~\eqref{eq:traceDist} becomes
	\begin{align*}
		Q_{Prod} (\sigma)
		&=
		p_m - \tilde{u}_m^2
		+ \sum_{i > m} p_i
		\\
		&=
		1 - \sum_{i < m} p_i
		- {\left(1 - \sum_{i < m} \sqrt{p_i}\right)}^2,
	\end{align*}
	which completes the proof.
\end{proof}

\section{Proof of Theorem~\ref{th:pure-decomposition}}\label{app:pure-decomposition}
Let us first gather the tools necessary to prove Theorem~\ref{th:pure-decomposition}.
We shall repeatedly use the following fact to compute the trace norm of $2 \times 2$ matrices.
	\begin{fact}\label{fact:norm}
		Suppose $A$ is a $2 \times 2$ Hermitian matrix
		\begin{align*}
			A
			&=
			a_0 \identity
			+ a_x \sigma_x
			+ a_y \sigma_y
			+ a_z \sigma_z,
		\end{align*}
		with $a_0, a_x, a_y, a_z \in \mathbb{R}$.

		Then $\norm{A} = 2 \max \left\{ |a_0|, \sqrt{a_x^2 + a_y^2 + a_z^2} \right\}$
	\end{fact}
We will use the inequality shown in the following lemma.
\begin{lemma}
	Let $\psi$ be a pure state.
	Let $\chi_S$ be a product state that is diagonal in the Schmidt basis of $\psi$.
	Then
	\begin{align*}
		d_T (\psi, \chi_S)
		&\geq
		N(\psi) + \half C_T(\psi).
	\end{align*}\label{lemma:schmidt}
\end{lemma}
\begin{proof}
	Let $\ket{\psi} = \sum_i \sqrt{p_i} \ket{ii}$ be the Schmidt decomposition of $\psi$.
	Then $\sigma = \sum_i p_i \proj{ii}$ is the closest classically-correlated state to $\psi$ (see Theorem~\ref{th:pure-cc}).
	Let $\chi_S = \sum_{ij} r_i s_j \proj{ij}$ be an arbitrary product state that is diagonal in the Schmidt basis of $\psi$.
	We show the inequality holds by direct computation.

	By definition, the partial transpose distance is $d_T (\psi, \chi_S) = \half \norm{\pt{\psi} - \pt{\chi_S}}$.
	Notice that $\pt{\psi}$ and $\pt{\chi_S}$ decompose into blocks $\{\Pi_{ii} = \proj{ii}\}$ and ${\{\Pi_{ij} = \proj{ij} + \proj{ji}\}}_{i < j}$, so we can sum the contribution from each block.
	For the $\Pi_{ii}$ blocks, we have
	\begin{align*}
		\bra{ii} \pt{\psi} - \pt{\chi_S} \ket{ii}
		&=
		p_i - r_i s_i,
	\end{align*}
	so their total contribution to the norm is
	\begin{align*}
		\sum_i \norm{\Pi_{ii} \left(\pt{\psi} - \pt{\chi_S}\right)}
		&= \sum_i | p_i - r_i s_i |.
	\end{align*}
	For the $\Pi_{ij}$ blocks, we obtain $2 \times 2$ submatrices
		\begin{align*}
		\Pi_{ij} \left( \pt{\psi} - \pt{\chi_S} \right)
		&=
		\begin{pmatrix}
			-r_i s_j & \sqrt{p_i p_j} \\
			\sqrt{p_i p_j} & -r_j s_i
		\end{pmatrix}
		\\
		&=
		- \left( \frac{r_i s_j + r_j s_i}{2} \right)
		\identity
		+ \sqrt{p_i p_j} \sigma_x
		- \left( \frac{r_i s_j - r_j s_i}{2} \right) \sigma_z.
	\end{align*}
	By Fact~\ref{fact:norm},
	\begin{align*}
		\norm{\Pi_{ij} \left( \pt{\psi} - \pt{\chi_S} \right)}
		&=
		2 \max \left\{
			\frac{r_i s_j + r_j s_i}{2},
			\sqrt{p_i p_j + {\left(\frac{r_i s_j - r_j s_i}{2}\right)}^2}
		\right\}
		\\
		&\geq
		2 \max \left\{
			\frac{r_i s_j + r_j s_i}{2},
			\sqrt{p_i p_j}
		\right\}.
	\end{align*}
The total norm reads
	\begin{align}
		\norm{\pt{\psi} - \pt{\chi_S}}
		&=
		\sum_i \norm{\Pi_{ii} \left( \pt{\psi} - \pt{\chi_S} \right)}
		+ \sum_{i < j} \norm{\Pi_{ij} \left( \pt{\psi} - \pt{\chi_S} \right)}
		\nonumber \\
		&\geq
		\sum_{i} \left| p_i - r_i s_i \right|
		+ \sum_{i < j} 2 \max \left\{ \frac{r_i s_j + r_j s_i}{2}, \sqrt{p_i p_j} \right\}
		\nonumber \\
		&=
		\sum_{ij} \delta_{ij} \left| \sqrt{p_i p_j} - \frac{r_i s_j + r_j s_i}{2} \right|
		+ \sum_{ij} (1 - \delta_{ij}) \max \left\{ \frac{r_i s_j + r_j s_i}{2}, \sqrt{p_i p_j} \right\}.
	\end{align}
Let $c_{ij} = \frac{r_i s_j + r_j s_i}{2}$.
	Using the identity $2\max\{a, b\} = |a+b| + |a-b|$ inside the second sum, we have
	\begin{align}
		\norm{\pt{\psi} - \pt{\chi}}
		&\geq
		\sum_{ij} \delta_{ij} | c_{ij} - \sqrt{p_i p_j} |
		+ \sum_{ij} (1 - \delta_{ij}) \left(
			\frac{c_{ij} + \sqrt{p_i p_j}}{2}
			+ \left| \frac{c_{ij} - \sqrt{p_i p_j}}{2} \right|
		\right)
		\nonumber \\
		&=
		\sum_{ij} (1 + \delta_{ij}) \left|
			\frac{c_{ij} - \sqrt{p_i p_j}}{2}
		\right|
		+ \sum_{ij} (1 - \delta_{ij}) \left(
			\frac{c_{ij} + \sqrt{p_i p_j}}{2}
		\right)
		\nonumber \\
		&=
		\sum_{ij} (1 + \delta_{ij}) \left|
		\frac{c_{ij} - \sqrt{p_i p_j}}{2}
		\right|
		+ \frac{1 - \sum_i r_i s_i}{2}
		+ \frac{{\left( \sum_i \sqrt{p_i} \right)}^2 - 1}{2}
		\nonumber \\
		&=
		\sum_{ij} \left|
			\frac{c_{ij} - \sqrt{p_i p_j}}{2}
		\right|
		+ \sum_{i} \left|
			\frac{r_i s_i - p_i}{2}
		\right|
		+ \frac{1 - \sum_i r_i s_i}{2}
		+ \frac{{\left( \sum_i \sqrt{p_i} \right)}^2 - 1}{2}.
	\end{align}
By direct computation, we have
	\begin{align*}
		\norm{\pt{\psi} - \pt{\sigma}}
		&=
		{\left( \sum_i \sqrt{p_i} \right)}^2 - 1,
		\\
		\norm{\pt{\sigma} - \pt{\chi_S}}
		&=
		\sum_i \left| r_i s_i - p_i \right| + 1 - \sum_i r_i s_i.
	\end{align*}
We thus have
	\begin{align}
		\norm{\pt{\psi} - \pt{\chi_S}}
		&\geq
		\norm{\pt{\psi} - \pt{\sigma}}
		+ \half \norm{\pt{\sigma} - \pt{\chi_S}}
		+ \sum_{ij} \left|
			\frac{c_{ij} - \sqrt{p_i p_j}}{2}
		\right|
		- \frac{{\left( \sum_i \sqrt{p_i} \right)}^2 - 1}{2}.
	\end{align}
In the next step, we show that $ \sum_{ij} \left| c_{ij} - \sqrt{p_i p_j} \right| + 1 - {\left( \sum_i \sqrt{p_i} \right)}^2 \geq 0$.
Since $\sum_{ij} c_{ij} = 1$, we have
	\begin{align*}
		\sum_{ij} \left| c_{ij} - \sqrt{p_i p_j} \right|
		+ 1 - {\left( \sum_i \sqrt{p_i} \right)}^2
		&=
		\sum_{ij} \left| c_{ij} - \sqrt{p_i p_j} \right|
		+ \sum_{ij} c_{ij}
		- \sum_{ij} \sqrt{p_i p_j}
		\\
		&= \sum_{ij}
		\left| c_{ij} - \sqrt{p_i p_j} \right|
		+ \left( c_{ij} - \sqrt{p_i p_j} \right)
		\\
		&= \sum_{ij}
		2 \max \{
			0,
			c_{ij} - \sqrt{p_i p_j}
		\}
		\\
		&\geq 0.
	\end{align*}
For any $\chi_S$ that is diagonal in the Schmidt basis, we therefore obtain
	\begin{align*}
		d_T (\psi, \chi_S)
		&\geq
		d_T (\psi, \sigma)
		+ \half d_T (\sigma, \chi_S).
	\end{align*}
	In the last step we note that the first term on the right-hand side is equal to $N(\psi)$ and the second term is the upper bound on $C_T(\psi)$ as the state $\chi_S$ might not be the closest one to $\sigma$. This completes the proof.
	\end{proof}

We move to the main theorem.

\begin{theorem*}
	Let $\psi$ be a pure state.
	Then
	\begin{align*}
		I_T(\psi)
		&\leq
		N(\psi) + \half C_T(\psi).
	\end{align*}
	If $\psi$ is maximally entangled, then we have equality.
\end{theorem*}
\begin{proof}
	Let us write $\psi$ in its Schmidt basis $\ket{\psi} = \sum_i \sqrt{p_i} \ket{ii}$ and let us define
	\begin{align}
		m
		&= \max{\left\{
				n\,|\, \sum_{i < n} \sqrt{p_i} \leq 1
		\right\}} \nonumber \\
		\chi
		&= {\left(
			\sum_{i < m} \sqrt{p_i} \proj{i}
			+ \left(1 - \sum_{i < m} \sqrt{p_i}\right) \proj{m}
		\right)}^{\otimes 2}
		\nonumber \\
		\sigma_0
		&= \sum_i p_i \proj{ii}.
	\end{align}
	Recall that Theorem~\ref{th:pure-cc} showed that $\sigma_0$ is the closest classically-correlated state to $\psi$
	and Theorem~\ref{th:cc-product} showed that $\chi$ is the closest product state to $\sigma_0$.
	We now prove a strict equality
	\begin{align}
		d_T (\psi, \chi)
		&=
		d_T (\psi, \sigma_0)
		+ \half d_T (\sigma_0, \chi).
		\label{EQ_EQSUB}
	\end{align}
	and the final inequality follows by noting that the state $\chi$ might not be the closest product state to the pure state $\psi$.

	First, notice that $\pt{\psi}$ decomposes as
	\begin{align*}
		\pt{\psi}
		&=
		\sum_i p_i \proj{ii}
		+ \sum_{i \neq j} \sqrt{p_i p_j} \ket{ij}\bra{ji},
	\end{align*}
	with blocks $\{ \Pi_{ii} = \proj{ii} \}$ and ${\{ \Pi_{ij} = \proj{ij} + \proj{ji} \}}_{i < j}$.
	Furthermore, $\chi$ is also diagonal in the basis $\{ \ket{ij} \}$, and $\pt{\chi} = \chi$.
	Thus to compute $\norm{\pt{\psi} - \pt{\chi}}$, we can compute the contributions from each block, and then sum them all up.

	First, we look at the $\Pi_{ii}$ blocks.
	We have
	\begin{align*}
		\bra{ii} \left( \pt{\psi} - \pt{\chi} \right) \ket{ii}
		&=
		p_i - \bra{ii} \chi \ket{ii}
		\\
		&=
		\begin{cases}
			0, & \text{for}\, i < m,
			\\
			p_m - {\left(1 - \sum_{i < m} \sqrt{p_i}\right)}^2, & \text{for}\, i = m,
			\\
			p_i, & \text{for}\, i > m.
		\end{cases}
	\end{align*}
	Since $\sqrt{p_m} > 1 - \sum_{i < m} \sqrt{p_i}$ by the definition of $m$, we must have $p_m - {\left( 1 - \sum_{i < m} \sqrt{p_i} \right)}^2 \geq 0$.
	So the total contribution to the norm from the $\Pi_{ii}$ blocks is
	\begin{align*}
		\sum_i \norm{\proj{ii} \left( \pt{\psi} - \pt{\chi} \right) \proj{ii}}
		&=
		\sum_{i \geq m} p_i - {\left(1 - \sum_{i < m} \sqrt{p_i}\right)}^2.
	\end{align*}

	In the $\Pi_{ij}$ blocks, we have the following sub-matrix
	\begin{align*}
		\Pi_{ij} \left( \pt{\psi} - \pt{\chi} \right)
		&=
		\begin{pmatrix}
			- \bra{ij} \chi \ket{ij}
			& \sqrt{p_i p_j}
			\\
			\sqrt{p_i p_j}
			& - \bra{ji} \chi \ket{ji}
		\end{pmatrix}.
	\end{align*}
	Since $\bra{ij} \chi \ket{ij}$ depends on $i, j$, we look at each case individually.
	\begin{itemize}
		\item $i < j < m$:

		We have $\bra{ij} \chi \ket{ij} = \bra{ji} \chi \ket{ji} = \sqrt{p_i p_j}$, and hence $\norm{\Pi_{ij} \left( \pt{\psi} - \pt{\chi} \right)} = 2 \sqrt{p_i p_j}$.

		\item $i < m, j = m$:

		We have $\bra{im} \chi \ket{im} = \bra{mi} \chi \ket{mi} = \sqrt{p_i} \left( 1 - \sum_{i < m} \sqrt{p_i} \right)$.
		The $\Pi_{im}$ sub-matrix reads
		\begin{align*}
			\Pi_{im} \left( \pt{\psi} - \pt{\chi} \right)
			&=
			\begin{pmatrix}
				- \sqrt{p_i} \left( 1 - \sum_{i < m} \sqrt{p_i} \right)
				& \sqrt{p_i p_m}
				\\
				\sqrt{p_i p_m}
				& - \sqrt{p_i} \left( 1 - \sum_{i < m} \sqrt{p_i} \right)
			\end{pmatrix}
			\\
			&= \sqrt{p_i}
			\begin{pmatrix}
				- \left( 1 - \sum_{i < m} \sqrt{p_i} \right)
				& \sqrt{p_m}
				\\
				\sqrt{p_m}
				& - \left( 1 - \sum_{i < m} \sqrt{p_i} \right)
			\end{pmatrix}.
		\end{align*}
		We get $\norm{\Pi_{im} \left( \pt{\psi} - \pt{\chi} \right)} = 2 \sqrt{p_i p_m}$.

		\item $i = m, j > m$:

		We have $\bra{mj} \chi \ket{mj} = \bra{jm} \chi \ket{jm} = 0$.
		Thus
		\begin{align*}
			\Pi_{mj} \left( \pt{\psi} - \pt{\chi} \right)
			&=
			\begin{pmatrix}
				0 & \sqrt{p_m p_j} \\
				\sqrt{p_m p_j} & 0
			\end{pmatrix},
		\end{align*}
		and $\norm{\Pi_{mj} \left( \pt{\psi} - \pt{\chi} \right)} = 2\sqrt{p_m p_j}$.

		\item $i > m, j > m$:

		We have $\bra{ij} \chi \ket{ij} = \bra{ji} \chi \ket{ji} = 0$, and $\norm{\Pi_{ij} \left( \pt{\psi} - \pt{\chi} \right)} = 2\sqrt{p_i p_j}$.

	\end{itemize}
	Summarizing, in all the cases we have $\norm{\Pi_{ij} \left( \pt{\psi} - \pt{\chi} \right)} = 2 \sqrt{p_i p_j}$.

	Therefore,
	\begin{align}
		\norm{\pt{\psi} - \pt{\chi}}
		&=
		\sum_i \norm{\Pi_{ii} \left( \pt{\psi} - \pt{\chi} \right)}
		+ \sum_{i < j} \norm{\Pi_{ij} \left( \pt{\psi} - \pt{\chi} \right)}
		\nonumber \\
		&=
		\sum_{i \geq m} p_i - {\left(1 - \sum_{i < m} \sqrt{p_i}\right)}^2
		+ \sum_{i < j} 2 \sqrt{p_i p_j}
		\nonumber \\
		&=
		1 - \sum_{i < m} p_i - {\left(1 - \sum_{i < m} \sqrt{p_i}\right)}^2
		+ \sum_{i < j} 2 \sqrt{p_i p_j}.
	\end{align}
	Note that $\sum_{i < j} 2 \sqrt{p_i p_j} = 2 d_T (\psi, \sigma_0)$ and $ 1 - \sum_{i < m} p_i - {\left(1 - \sum_{i < m} \sqrt{p_i}\right)}^2 = d_T (\sigma_0, \chi)$.
	We therefore obtain Eq.~\eqref{EQ_EQSUB}.

	To prove equality for maximally entangled states, our strategy is to use the result of Lemma~\ref{lemma:schmidt}.
	It shows that the distance from $\psi$ to a product state diagonal in the Schmidt basis is lower bounded by some combination of negativity and classical correlations in $\psi$.
	On the other hand, the theorem shows that the distance from $\psi$ to its closest product state is upper bounded by the same combination of negativity and classical correlations.
	Accordingly, showing that the closest product state is diagonal in the Schmidt basis proves the anticipated equality.

	We proceed as follows. Since $\psi$ is maximally entangled, there exists local unitaries $U, V$ such that $( U \otimes V ) \psi = \Phi$, where $\ket{\Phi} = \sum_i \frac{1}{\sqrt{d}} \ket{ii}$.
	The partial transpose distance is invariant under local unitaries, so we must have $I_T ( \psi ) = I_T (\Phi)$.
	Suppose $\chi = \chi_A \otimes \chi_B$ is the closest product state to $\Phi$.
	Since $( U \otimes U^* ) \Phi = \Phi$ for any local unitary $U$, we perform operation $U \otimes U^*$ that makes $\chi_A = \sum_i a_i \proj{i}$ diagonal in the Schmidt basis of $\Phi$. The closest product state written in the Schmidt basis now reads:
	\begin{align*}
		\chi^{T_B}
		&=
		\sum_{ijk} a_i b_{jk} \ket{ij} \bra{ik}.
	\end{align*}
	We show a projective map $\Pi$ that preserves the partial transposition of $\Phi$, dephases the closest product state in the Schmidt basis, and preserves the product structure.
	The contractivity of trace norm~\cite{RUSKAI_1994} under positive trace-preserving maps then shows that the distance to $\Phi$ cannot increase:
	\begin{align}
	I_T (\Phi) = \frac{1}{2} || \Phi^{T_B} - \chi^{T_B} ||_1 \ge \frac{1}{2} || \Pi(\Phi^{T_B}) - \Pi(\chi^{T_B}) ||_1 = \frac{1}{2} || \Phi^{T_B} - \Pi(\chi^{T_B}) ||_1.
	\end{align}
	Accordingly, there always exists a closest product state that is diagonal in the Schmidt basis.

	The relevant projectors are as follows: $\Pi_{ii} = \proj{ii}$ and $\Pi_{ij} = \proj{ij} + \proj{ji}$.
	Indeed the matrix $\Phi^{T_B}$ is preserved by this operation whereas the product state remains product and is dephased:
	\begin{align}
		\sum_i \Pi_{ii} (\chi^{T_B})
		+ \sum_{i < j} \Pi_{ij} (\chi^{T_B})
		&=
		\left(
			\sum_i a_i \proj{i}
			\right) \otimes \left(
			\sum_j b_{jj} \proj{j}
		\right).
	\end{align}

	This completes the proof.
\end{proof}

\bibliographystyle{apsrev4-2}
\bibliography{refs.bib}

\end{document}